\documentclass[10pt, conference, compsocconf]{IEEEtran}
\usepackage{balance}
\usepackage{latexsym}
\usepackage{amsmath}
\usepackage{amsthm}
\usepackage{amssymb}
\usepackage{verbatim}
\usepackage{setspace}
\usepackage{algorithm}
\usepackage[noend]{algpseudocode}
\usepackage{hyperref, url}
\usepackage{ifthen}
\usepackage{graphicx, color, xcolor}
\usepackage{subcaption}
\usepackage{relsize}
\usepackage{tikz}

\usepackage[export]{adjustbox}
\usepackage{placeins}

\usepackage[utf8]{inputenc}
\usepackage{pgfplots}
\usepackage{colortbl}
\usepackage{caption}
\usepackage{graphicx}
\usepackage{booktabs}
\usepackage{multirow}
\usepackage{lmodern}% http://ctan.org/pkg/lm
\usepackage{tikz}

 \usepackage{mathtools}

\DeclarePairedDelimiter{\floor}{\lfloor}{\rfloor}

\newtheorem{theorem}{Theorem}
\newtheorem{lemma}[theorem]{Lemma}

\newtheorem{definition}{Definition}

\newcommand{\pd}{\mathbf{p}}

\newcommand{\med}{\text{mid}}

\newboolean{short}
\setboolean{short}{false}
\newcommand{\shortOnly}[1]{\ifthenelse{\boolean{short}}{#1}{}}
\newcommand{\onlyShort}[1]{\ifthenelse{\boolean{short}}{#1}{}}
\newcommand{\longOnly}[1]{\ifthenelse{\boolean{short}}{}{#1}}
\newcommand{\onlyLong}[1]{\ifthenelse{\boolean{short}}{}{#1}}
\newcommand{\shortLong}[2]{\ifthenelse{\boolean{short}}{#2}{#1}}
\newcommand{\longShort}[2]{\ifthenelse{\boolean{short}}{#2}{#1}} 
\onlyShort{
\let\labelindent\relax
\usepackage{enumitem}
\setitemize{noitemsep,topsep=0pt,parsep=0pt,partopsep=0pt}
\setenumerate{noitemsep,topsep=0pt,parsep=0pt,partopsep=0pt}
}

\newcommand{\anis}[1]{{\color{blue}\bf [Anisur: #1]}}

\def\eps{{\epsilon}}

\algblock[phase]{Phase}{EndPhase}
\algblock[initialize]{Init}{EndInit}
\algblock[receive]{OnReceive}{EndReceive}
\algblock[foreach]{foreach}{endForeach}

\usepackage[skip=-5pt]{subcaption}
\makeatletter
\ifthenelse{\equal{\ALG@noend}{t}}%
{\algtext*{EndInit}}
{}%
\makeatother

\makeatletter
\ifthenelse{\equal{\ALG@noend}{t}}%
{\algtext*{EndReceive}}
{}%
\makeatother

\makeatletter
\ifthenelse{\equal{\ALG@noend}{t}}%
{\algtext*{endForeach}}
{}%
\makeatother

\makeatletter
\renewcommand{\Function}[2]{%
	\csname ALG@cmd@\ALG@L @Function\endcsname{#1}{#2}%
	\def\jayden@currentfunction{#1}%
}

\newcommand{\funclabel}[1]{%
	\@bsphack
	\protected@write\@auxout{}{%
		\string\newlabel{#1}{{\jayden@currentfunction}{\thepage}}%
	}%
	\@esphack
}
\makeatother

\newcommand\blfootnote[1]{%
  \begingroup
  \renewcommand\thefootnote{}\footnote{#1}%
  \addtocounter{footnote}{-1}%
  \endgroup
}

\long\gdef\boxit#1{\vspace{5mm}\begingroup\vbox{\hrule\hbox{\vrule\kern3pt
			\vbox{\kern3pt#1\kern3pt}\kern3pt\vrule}\hrule}\endgroup}
		
\begin{document}

\title{Efficient Distributed Community Detection in the Stochastic Block Model}

\author{
\IEEEauthorblockN{Reza Fathi}
\IEEEauthorblockA{Department of Computer Science\\
  University of Houston \\
   Houston, Texas 77204, USA\\ 
Email: rfathi@uh.edu}

\and
\IEEEauthorblockN{Anisur Rahaman Molla}
\IEEEauthorblockA{Cryptology and Security Research Unit\\ 
Indian Statistical Institute \\
  Kolkata 700108, India \\
  Email: molla@isical.ac.in}

\and
\IEEEauthorblockN{Gopal Pandurangan}
\IEEEauthorblockA{Department of Computer Science\\
  University of Houston \\
   Houston, Texas 77204, USA\\ 
Email: gopalpandurangan@gmail.com}}

%\iffalse
%\author{
%	\IEEEauthorblockN{Reza Fathi\IEEEauthorrefmark{1},  Anisur Rahaman Molla\IEEEauthorrefmark{2}, and Gopal Pandurangan\IEEEauthorrefmark{1}}
%	\IEEEauthorblockA{\IEEEauthorrefmark{1}Department of Computer Science, University of Houston, Texas, USA\\
%	Email: rfathi@uh.edu, gopalpandurangan@gmail.com}
%	\IEEEauthorblockA{\IEEEauthorrefmark{2}Cryptology and Security Unit, Indian Statistical Institute, Kolkata, India\\ 
%Email: molla@isical.ac.in}
%}
%
%
%\section*{TODO List}
%\begin{itemize}
%\item Say the we interchangeably use ``stochastic block model" and ``planted partition model"
%\item In Intro, related work secs (and may be in other places also), discussion on CDST algorithm is there. We should remove them. Check throughout the paper and clean it.   
%\item \anis{IMPORTANT: Specify what is the main difference with the mixing time paper and current paper}. In this paper we are computing the set (in fact, the largest set) where the random walk mixing locally. In \cite{MP18}, the authors consider only the local mixing time which is essentially the existence of a mixing set of certain size, but not the set of nodes where the random walk mixes. The computation of local mixing set is expensive in distributed setting.     
%\item Say that we consider the CONGEST model of distributed computing. (Anisur--I added it briefly in the model section)
%\item The experimental results are missing (Reza?)   
%\item \anis{Define high probability somewhere in the beginning}  
%\end{itemize}
%
%\newpage
%\fi
\maketitle

\vspace{-0.2in}
\begin{abstract}
	Designing effective algorithms for community detection is an important and challenging problem in {\em large-scale} graphs, studied extensively in the literature. Various solutions have been proposed, but many of them
	are centralized with expensive procedures (requiring full knowledge of the input graph) and have a large running time. 
	
		In this paper,  we present a  distributed algorithm for community detection  in the {\em stochastic block model} (also called {\em planted partition model}), a widely-studied and canonical random graph model for community detection and clustering.
	Our  algorithm called {\em CDRW(Community Detection by Random Walks)} 
	 is based on random walks, and is localized and lightweight, and  easy to implement.
	 A novel feature of the algorithm is that it uses the concept of {\em local mixing time} to identify the community around a given node. 
	
	 We present a rigorous theoretical analysis that shows that the algorithm can accurately identify the communities in the stochastic block model and characterize the model parameters where the algorithm works. We also present experimental results that validate our theoretical analysis. 
	 We also analyze the performance of our distributed algorithm under the CONGEST distributed model
	 as well as the $k$-machine model, a model for large-scale distributed computations, and show that it can be efficiently implemented.
	 \end{abstract}

\vspace{-0.4cm}
\blfootnote{A. R. Molla was supported by DST Inspire Faculty Award research grant DST/INSPIRE/04/2015/002801.}
\blfootnote{G. Pandurangan was supported, in part, by NSF grants CCF-1527867, CCF-1540512,  IIS-1633720,  CCF-1717075, and
BSF award 2016419.}
%\blfootnote{A Long version of this paper is available at: https://sites.google.com/site/csrfathi.}

\IEEEpeerreviewmaketitle


\section{introduction}
\label{sec:intro}
Finding communities in networks (graphs) is an important problem and has been extensively studied in the last two decades, e.g., see the surveys \cite{abbe2018community,clementi-tcs15,fortunato2010community,fortunato2016community} and other references in  Section \ref{sec:related}.
At a high level, the goal is to identify subsets of vertices of the given graph so that each subset represents a ``community." 
While there are differences in how communities are defined exactly (e.g., subsets defining a community may overlap or not), a uniform property that underlies most definitions
is that there are more edges  connecting nodes {\em within} a subset than edges connecting to {\em outside} the subset. Thus, this captures the fact that a community is somehow more ``well-connected" with respect to itself compared to the rest of the graph. One way to express this property formally is by using the notion of {\em conductance}  (defined in Section \ref{sec:local-mixing}) which (informally) measures the ratio of the total degree going out of the subset to the total degree among all nodes in the subset. A community subset will have generally low conductance (also sometime referred to as a ``sparse" cut). Another way to measure this is by using the notion of {\em modularity} \cite{newman2004finding} which (informally) measures how well connected a subset is compared to the random graph that can be embedded within the set. A vertex subset that has a  high modularity value can be considered a community according to this measure.

Designing effective algorithms for community detection in graphs is an important and challenging problem. With the rise
of massive graphs such as the Web graph, social networks, biological networks, finding communities (or clusters) in large-scale graphs efficiently is becoming even more  important \cite{clementi-tcs15,easley2010networks, fortunato2016community,girvan2002community,newman2004finding}. In particular, understanding the community structure  is
a key issue in many applications in various complex networks including biological and social networks (see e.g., \cite{clementi-tcs15} and the references therein).

Various solutions have been proposed (cf. Section  \ref{sec:related}), but many of them
	are centralized  with expensive procedures (requiring full knowledge of the input graph) and have a large running time \cite{clementi-tcs15}. In particular, the problem of detecting (identifying) communities in the {\em stochastic block model (SBM)} \cite{abbe2018community,clementi-tcs15}
	has been extensively studied in the literature (cf. Section \ref{sec:related}). The stochastic block model (defined formally in Section \ref{sec:model}), also
	known as the {\em planted partition model (PPM)} is a widely-used random graph model in community detection and clustering studies (see e.g.,  \cite{abbe2018community,chin2015stochastic,karrer2011stochastic}).  Informally, in the PPM model, we are given a graph $G$ on $n$ nodes which are partitioned into
	a set of  $r$ communities (each is a random graph on $n/r$ vertices) and these communities are interconnected by random edges. The total number of  edges within a community (intra-community edges) is typically much more than the number of edges between communities (inter-community edges). The main goal is to devise algorithms to identify the $r$ communities that are ``planted" in the graph.  Several algorithms have been devised for this model, but as mentioned earlier, they are mostly centralized (with some exceptions --- cf. Section \ref{sec:related}) with large running time. There are few distributed algorithms  (see e.g., \cite{clementi-tcs15,kothapalli2013analysis}) but either they are shown to work only when the number of communities is small (typically 2) \cite{clementi-tcs15}) or when the communities are very dense\cite{kothapalli2013analysis}).   In particular, to the best of our knowledge,  (prior to this work) there is no  rigorous analysis of community detection algorithms in distributed large-scale graph processing models
	such as MapReduce and Massively Parallel Computing (MPC) model~\cite{soda-mapreduce}, 
	and the $k$-machine model \cite{KlauckNPR15}.  %Some distributed algorithms have been proposed for the PPM model, but they mostly

\iffalse
Various solutions  have been proposed, but not a single solution works well on all types of graphs. For example, some solutions work well for  sparse graphs and others work well on denser graphs. Similarly
some work on graphs where the size (as well as the edge density) of the communities are more less the same verses graphs where the communities sizes (and edge densities) vary widely.
\fi
	 \vspace{-0.1in}
	\subsection{Our Contributions}
	\vspace{-0.05in}
	In this paper, we present a  novel distributed algorithm, called {\em CDRW} ({\em Community Detection via Random Walks}) for detecting  communities in the PPM model\cite{condon2001algorithms}.\footnote{Throughout this paper, we use the terms stochastic block model (SBM) and planted partition model (PPM) interchangeably.} 
	Our algorithm 
	 is based on the recently proposed  {\em local mixing}  paradigm \cite{MP18} (see  Section \ref{sec:local-mixing} for  a formal definition) to detect community structure in {\em sparse} (bounded-degree) graphs.  Informally, a local mixing set is one where a random walk started at some node in the set mixes well {\em with respect
	to this set.} The intuition in using this concept for community detection is that since a community  is well-connected, it has good expansion within the community
	and hence a random walk started at a node in the community mixes well within the community. The notion of ``mixes well" is captured by the fact
	that the random walk reaches close to the stationary distribution  when {\em restricted to the nodes in the community subset} \cite{MP18}.  
	Since the main tool for this algorithm uses random walks which is local and lightweight, it is easy to implement this algorithm in a distributed manner.  We will  analyze the performance of the algorithm in two distributed computing models,
	namely the standard CONGEST model of distributed computing \cite{peleg} and the $k$-machine model \cite{KlauckNPR15}, which is a model
	for large-scale  distributed computations. We show that CDRW can be implemented efficiently in both models  (cf. Theorem \ref{thm:r-community}
	and Section \ref{subsec:kmachine}). The $k$-machine model implementation is especially suitable for large-scale graphs and thus can be used
	in community detection in large SBM graphs. In particular, we show that the round complexity in the $k$-machine model
	(cf. Section \ref{subsec:kmachine}) scales  quadractically (i.e., $k^{-2}$) in the number of machines when the graph is sparse and linearly (i.e.,$k^{-1}$)  in general.
	
	As is usual in community detection, a main focus is analyzing the effectiveness of the algorithm in finding communities.
	 We present a rigorous theoretical analysis that shows that CDRW algorithm can accurately identify the communities in the PPM, which is a popular and widely-studied random graph model for community detection analysis \cite{abbe2018community}. A PPM model (cf. Section \ref{sec:model}) is a parameterized random graph model which has a built-in community structure. Each community has high expansion within the community and forms a low conductance subset
	 (and hence relatively less edges go outside the community); the  expansion,  conductance, and edge density can be controlled by varying the parameters. CDRW does well when the number of intra-community edges is much more than
	 the number of inter-community edges (these are controlled by the parameters of the model).
	Our theoretical analysis  (cf. Theorem \ref{thm:r-community}  for the precise statement) quantitatively characterizes when CDRW does well vis-a-vis the parameters of the model. Our results improve over previous distributed algorithms that have been proposed
	for the PPM model  (\cite{clementi-tcs15}) both in the number of communities that can be provably detected as well
	as range of parameters where accurate detection is possible; they also improve
	on previous results that  provably work only on dense PPM graphs \cite{kothapalli2013analysis} (details in Section \ref{sec:related}).
	
	We also present extensive simulations of our algorithm in the PPM model under various parameters. Experimental results on the model validate our theoretical analysis; in fact experiments show that CDRW works relatively well in identifying communities even under less stringent parameters. 
	 %However, CDRW doesn't work well for real-world graphs which don't have uniform degree distribution and often have a varying and not so clear-cut community structure with  varying sizes and edge densities.

	 %Our work is a step towards designing  community detection algorithms that are effective vis-a-vis the underlying graph structure.

%\section{introduction}\label{sec:intro}
%\section{Preliminaries}
\vspace{-0.1in}
\subsection{Model, Definitions, and Preliminaries}\label{sec:model}
\vspace{-0.05in}
\subsubsection{Graph Model}
\label{sec:ppm}
We describe the  {\em stochastic block model (SBM)} \cite{holland1983stochastic}, a well-studied random graph model that is used in community detection analysis.\onlyLong{ Before that we recall the {\em Erd\H{o}s-R{\'e}nyi} random graph model \cite{erds1960evolution}, a basic random graph model, that the SBM model builds on.} In the Erd\H{o}s-R{\'e}nyi random graph model, also known as the $G_{np}$ model, each  of $\binom{n}{2}$ possible edges is present in the graph independently with probability $p$. The $G_{np}$ random graph has close to uniform degree (the expected degree of a node is $(n-1)p$) and a well-known fact about $G(n,p)$ is that
if $p = \Theta(\log n/n)$ the graph is connected with high probability\footnote{Throughout this paper, by ``with high probability" we mean, ``with probability at least $1 -1/n$."} and its degree is concentrated at its expectation. 
%Nodes communicate through the edges in synchronized rounds; in every round, every node is allowed to send a message of size at most $O(\log(n))$ bits to each of its neighbors. The message will arrive to the receiver nodes at the end of the current round. Initially every node has limited knowledge; it only knows its ID and its neighbors IDs.  
%Let us describe two well known random graph models that have ground truth community structure. The graphs will be used as benchmark to analyze and evaluate the performance of one of our algorithm. 
%We consider the random regular expander graph $G_{np}$, where The $G_{np}$ graph is well studied in the context of random graph model and also known as $Erd\H{o}s-Renyi\ random\ graph$ \cite{erds1960evolution}. The $G_{np}$ graph has a uniform degree distribution and being considered as a single community. 
%Another class of random graphs with ground truth communities is stochastic block model (SBM) \cite{holland1983stochastic}. 
\onlyLong{
In the SBM model, the vertices are partitioned into $r$ disjoint community sets ($r$ is a parameter). Let $c(u)$ denote the  community that node $u$ belongs to.
If two nodes $u$ and $v$ belong to the same community (i.e., $c(u) = c(v)$), they connect independently with probability $p_{c(u)}$ (independent of other nodes). If they belong to different communities, they connect independently  with probability $p_{c(u),c(v)}$. A SBM model has a separable community structure \cite{abbe2018community} if it has a higher value of intra- than inter-community connectivity probability.
A commonly used version of SBM model that is called }\onlyShort{\\} {\em Planted Partition Model (PPM)}, denoted as $G_{npq}$ \cite{condon1999algorithms,holland1983stochastic} is usually used as a benchmark. A symmetric $G_{npq}$ with $r$ blocks is composed of $r$ exclusive set of nodes in $\mathcal{C}=<C_1,C_2,\dots,C_r>$, where $\cup_{i=1}^{i=r} C_i=V,\ |C_i|=|C_j|  \text{ and }C_i \cap C_j=\emptyset \text{ when } i\neq j$. In $G_{npq}$, each possible edge $e(u,v)$ is independently present with probability $p$ if both ends of $u$ and $v$ belong to the same block $C_i$, otherwise with probability $q$. \onlyLong{Figure \ref{fig:graphcom} shows a PPM of $5$ blocks ($r=5$), each containing $200$ nodes.}
For the PPM model, each of the $r$ communities induce a $n/r$-vertex subgraph which is a $G(n/r,p)$ random graph. The goal of community detection in the PPM (or SBM) model
is to identify the $r$ community sets. This problem has been widely studied \cite{abbe2018community}.

\iffalse
\subsubsection{Notion of a Community}
\label{sec:community}
While there is not a universally accepted definition for communities \cite{fortunato2016community}, in this paper, we use a reasonable definition
that is either the same or closely related to what is used in the literature (cf. see Section \ref{sec:related}).  Informally, we define a community as a subset of nodes ($C_i \subset V$) where nodes inside the subset are ``more connected" compared to nodes outside of the community. In other words, for each $v\in C_i$ it is $\sum_{u \in C_i} I_{e(v,u)}>\sum_{u \in V/C_i} I_{e(v, u)}$ where $I$ is an indicator function and returns $1$ if an edge $e(v, u)$ exists, otherwise $0$. In this paper, we look for non-overlapping communities of an undirected graph $G(V,E)$. Let $\mathcal{C}=\{C_1, C_2, ..., C_r\}$, we call it non-overlapping if $\cup C_i = V$ and $C_i \cap C_j=\emptyset$ when $i\neq j$. For the SBM, the communities are the $r$ subsets of the vertices that
induce a random graph. In particular, for the PPM model, each of the $r$ communities induce a $n/r$-vertex subgraph which is a $G(n/r,p)$ random graph. The goal of community detection in the PPM (or SBM) model
is to identify the $r$ community sets. This problem has been widely studied \cite{abbe2018community}.
\fi

\onlyLong{
\begin{figure}
	\begin{subfigure}[b]{0.5\columnwidth}
		\includegraphics[width=\textwidth]{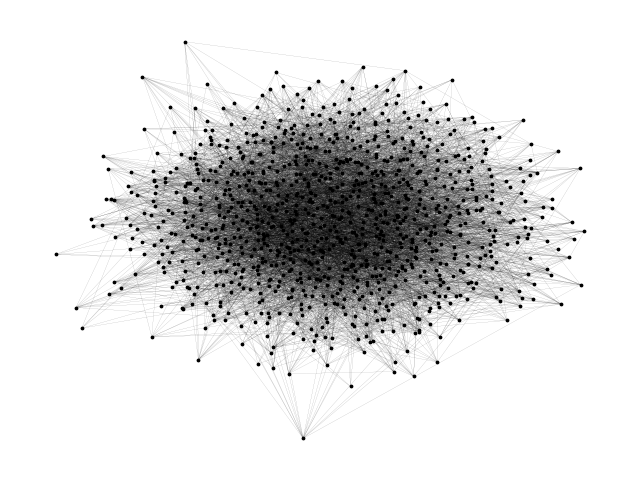}
		\caption{A PPM graph shown without its ground-truth communities.}
		\label{fig:org}
	\end{subfigure}%
	~
	\begin{subfigure}[b]{0.5\columnwidth}
		\includegraphics[width=\textwidth]{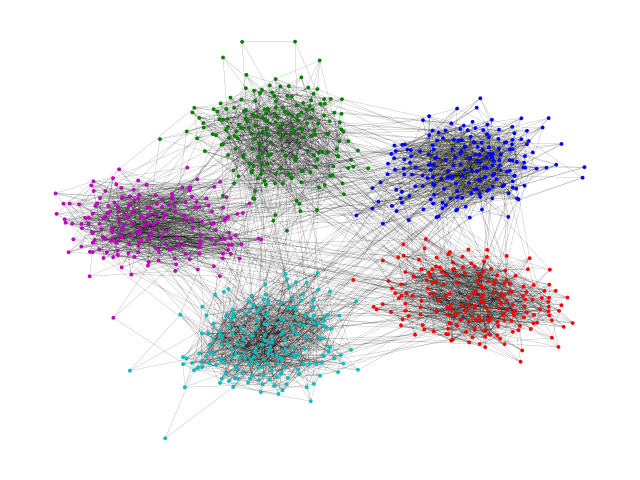}
		\caption{Redrawing of the same PPM graph showing the  communities.}
		\label{fig:blocks}
	\end{subfigure}
	\caption{\footnotesize Both of the above graphs are the same PPM graph. The graph size is $n=1000$, number of communities $r=5$, the existence probability of intra(inter) community edges is $p=\frac{1}{20}$($q=\frac{1}{1000})$ as in \cite[Figure~1]{abbe2018community}. We highlight ground truth communities in different colors in Figure\ref{fig:blocks}.
	}
	\label{fig:graphcom}
\end{figure}
}

\subsubsection{Distributed Computing Models}
\label{sec:dist-model}
We consider two  distributed computing models and analyze the complexity of the algorithm's implementation under both the models.

\noindent {\bf CONGEST model.} This is a standard and widely-studied {\em CONGEST} model of distributed computing \cite{gopal-book,peleg}, which captures the bandwidth constraints inherent in real-world networks.  
The distributed network is modelled as an undirected and unweighted graph $G = (V, E)$, where $|V| = n$ and $|E| = m$,
where nodes represent processors (or computing entities) and the edges represent communication links. In this paper, $G$ will
be a graph belonging to the PPM model. Each node is associated with a distinct label or ID (e.g., its IP address).  Nodes communicate through the edges in synchronous rounds; in every round, each node is allowed to send a message of size at most $O(\log n)$ bits to each of its neighbors. The message will arrive to the receiver nodes at the end of the current round. Initially every node has limited knowledge; it only knows its own ID and its neighbors IDs. In addition, it may known additional parameters of the graph such as $n,m,D$ (where $D$ is the diameter). 
%In addition we assume that nodes know that the network is an SBM graph  (but do not know the parameters $p$ and $q$).  
The two standard complexity of an algorithm are
the time and message complexity  in the CONGEST model. While time complexity measures
the number of rounds taken by the algorithm, the message complexity measures the total number of messages exchanged during the course of the algorithm.

\noindent {\bf $k$-machine model.} The $k$-machine model (a.k.a. Big Data model) is a  model for {\em large-scale} distributed computing introduced in \cite{KlauckNPR15} and studied in various papers \cite{BandyapadhyayIPP18,pemmaraju-opodis18,KlauckNPR15,spaa16,PanduranganRS18}.
 In this model, the input graph (or more generally, any other type of data)  is distributed across a group of
$k \geq 2$ machines that are pairwise interconnected via a communication network. 
The $k$ machines jointly perform computations on an arbitrary $n$-vertex, $m$-edge input graph (where typically $n,m \gg k$) distributed among the machines (randomly or in a balanced fashion). 
%There are two common ways of distributing the input graph
%that have been considered: {\em vertex partition} and {\em edge partition} (more on this below).
%which is assumed to be initially randomly partitioned among the $k$ machines, a common
%implementation in many real world graph processing systems~\cite{pregel,stanton,Stanton14,facebook}.
 The communication
is point-to-point via message passing. Machines do not share any memory and have no other
means of communication. 
  The computation advances in synchronous rounds, and
each link is assumed to have
a bandwidth of $B$ bits per round, i.e., $B$ bits can be
transmitted over each link in each round;  unless otherwise stated, we assume $B = O(\log n)$ (where $n$ is the input size) \cite{spaa16,PanduranganRS18}.
The goal is to minimize the {\em round  complexity},
i.e., the number of \emph{communication rounds} required by the computation.\footnote{The
communication cost is assumed to be the dominant cost -- which is typically the case in Big Data computations --- and hence the goal of minimizing the number of communication rounds \cite{cacm}.}

Initially, the entire graph $G$  is not known by any single machine, but rather partitioned among the $k$ machines in a
``balanced'' fashion, i.e., the nodes and/or edges of $G$ are partitioned approximately evenly among the machines.
We assume a {\em  vertex-partition} model, whereby vertices, along with information of their incident edges,
are partitioned across machines. Specifically, the type of partition that we will assume throughout is the
{\em random vertex partition (RVP)}, that is, each vertex of the input graph is assigned randomly to one
machine. (This is the typical way used by many real systems, such as Pregel~\cite{pregel}, to initially distribute the input graph
among the machines.)
%More formally, in the RVP model,  each vertex of $G$ is assigned independently and uniformly at
%random to one of the $k$ machines. 
If a vertex $v$ is assigned to machine $M_i$ we say that $M_i$ is the {\em home machine}
of $v$. % and, with a slight abuse of notation, write $v \in M_i$. 
%When a vertex is assigned to a machine, all its
%incident edges are assigned to that machine as well; i.e., the home machine knows the {IDs
%of the neighbors of that vertex as well as the identity of the home machines of the neighboring vertices (and
%the weights of the corresponding edges in case $G$ is weighted).
A convenient way to implement  the RVP model is through hashing: each vertex (ID)
is hashed to one of the $k$ machines. Hence, if a machine knows a vertex ID, it also knows where it is hashed to. 
It can be shown that the RVP model results in (essentially) a {\em balanced} partition of the graph: each machine gets $\tilde{O}(n/k)$ vertices and $\tilde{O}(m/k + \Delta)$ edges, where $\Delta$ is the maximum degree. 

\iffalse
Note that we can also assume an
alternate partitioning model, the \emph{random edge partition (REP)} model, where each edge of $G$ is
assigned independently and randomly to one of the $k$ machines. The results in the random
vertex partition model can be related to the random edge partition model \cite{spaa16}.
\fi

At the end of the computation, %for output, 
%each machine $M_i$ must set a designated local output variable $o_i$ (which need
%not depend on the set of vertices assigned to $M_i$), and the \emph{output configuration} $o=\langle o_1,\dots,o_k\rangle$
%must satisfy the feasibility conditions of the problem at hand. 
 for the community detection  problem,
the community that each vertex belongs to will be output by some machine.

\vspace{-0.1in}
\subsection{Random Walk Preliminaries and Local Mixing Set} 
\label{sec:local-mixing}
\vspace{-0.05in}
Our  algorithm is based on the mixing property of a random walk in a graph.  We use the notion of local mixing set of a random walk, introduced in \cite{MP18}, to identify communities in a graph. Let us define random walk preliminaries, local mixing time, and local mixing set as defined in \cite{MP18}. 
Given an undirected graph and a source node $s$, a {\em simple random walk} is defined as: in each step, the walk goes from the current node to a random neighbor i.e., from the current node $u$, the probability of moving to node $v$ is $\Pr(u, v) = 1/d(u)$ if $(u,v) \in E$, otherwise $\Pr(u, v) = 0$, where $d(u)$ is the degree of $u$. Let $\pd_t (s)$ be the probability distribution at time $t$ starting from the source node $s$. Then $\pd_0 (s)$ is the initial distribution with probability $1$ at the node $s$ and zero at all other nodes. The $\pd_t (s)$ can be seen as the matrix-vector multiplication between $(A)^t$ and $\pd_0(s)$, where $A$ is the transpose of the transition matrix of $G$. %Let the probability a node $v$ after $t$ steps is $p_t(s, v)$. 
Let $p_t(s, v)$ be the probability that the random walk be in node $v$ after $t$ steps.
When it's clear from the context we omit the source node from the notations and denote it by $p_t(v)$ only. The stationary distribution (a.k.a steady-state distribution) is the probability distribution which doesn't change anymore (i.e., it has converged). The stationary distribution of an undirected connected graph is a well-defined quantity which is $\bigl(\frac{d(v_1)}{2m}, \frac{d(v_2)}{2m}, \ldots, \frac{d(v_n)}{2m}\bigr)$, where $d(v_i)$ is the degree of node $v_i$.  %It is known that for every undirected connected graph $G$, the distribution $\pmb{\pi}(v) = d(v)/2m$ is stationary. 
We denote the stationary distribution vector by $\pmb{\pi}$, i.e., $\pi(v) = d(v)/2m$ for each node $v$. The stationary distribution of a graph is fixed irrespective of the starting node of a random walk, however, the number of steps (i.e., time) to reach to the stationary distribution could be different for different starting nodes. The time to reach to the stationary distribution is called the {\em mixing time} of a random walk with respect to the source node $s$. The mixing time corresponding to the source node $s$ is denoted by $\tau^{mix}_s$. The mixing time of the graph, denoted by $\tau^{mix}$, is the maximum mixing time among all (starting) nodes in the graph. %The formal definitions are given below. 

%The {\em mixing time} of a random walk starting from source node $s$, denoted by $\tau_s$, is the  time (or number of steps) taken to reach to the stationary distribution of the graph.  The mixing time, denoted by $\tau$, is the maximum mixing time among all (starting) nodes in the graph.

\begin{definition}\label{def:mixing-time} ($\tau^{mix}_s(\eps)$--mixing time for source $s$ and $\tau^{mix}(\eps)$--mixing time of the graph)\\
Define $\tau^{mix}_s (\eps)= \min \{t : ||\pd_t - \pmb{\pi}||_1 < \eps\}$, where $||\cdot||_1$ is the $L_1$ norm. Then $\tau^{mix}_s(\eps)$ is called the $\eps$-near mixing time for any $\eps$ in $(0, 1)$. The mixing time of the graph is denoted by $\tau^{mix} (\eps)$ and is defined by $\tau^{mix}(\eps) = \max \{\tau^{mix}_v(\eps): v \in V\}$. It is clear that $\tau^{mix}_s (\eps) \leq \tau^{mix} (\eps)$. \hfill $\blacksquare$
\end{definition}

%While a standard definition of mixing time usually takes $\eps$ to be $1/2e$ \cite{Levin}, we assume $\eps = o(1)$ throughout the paper.  
We sometime omit $\eps$ from the notations when it is understood from the context.  For any set $S\subseteq V$, we define $\mu(S)$ is the {\em volume} of $S$ i.e., $\mu(S) = \sum_{v \in S} d(v)$. Therefore, $\mu(V) = 2m$ is the volume of the vertex set. The {\em conductance} of the set $S$ is denoted by $\phi(S)$ and defined by 
$\phi(S) =  \frac{|E(S, V\setminus S)|}{\min\{\mu(S),\, \mu(V\setminus S)\}},$
where $E(S, V\setminus S)$ is the set of edges between $S$ and $V\setminus S$. The conductance of the graph $G$ is $\Phi_G = \min_{S\subseteq V} \phi(S)$. 

Let us define a vector $\pmb{\pi}_{S}$ over the set of vertices $S$ as follows: 
$\pi_{S}(v) = d(v)/\mu(S)$ if $v \in S$, and $\pi_{S}(v)=  0$ otherwise.

%\[ \pi_{S}(v) =
%  \begin{cases}
%    d(v)/\mu(S)       & \quad \text{if } v \in S\\
%    0  & \quad \text{otherwise} \\
%  \end{cases}
%\]

Notice that $\pmb{\pi}_V$ is the stationary distribution $\pmb{\pi}$ of a random walk over the graph $G$, and $\pmb{\pi}_S$ is the restriction of the distribution $\pmb{\pi}$ on the subgraph induced by the set $S$. Recall that we defined $\pd_t$ as the probability distribution over $V$ of a random walk of length $t$, starting from some source vertex $s$. Let us denote the restriction of the distribution $\pd_t$ over a subset $S$ by $\pd_t {\restriction_S}$ and define it as: 
$p_t {\restriction_S} (v) = p_t(v)$  if $v \in S$ and $p_t {\restriction_S} (v) =0$ otherwise.

%\[ p_t {\restriction_S} (v) =
%  \begin{cases}
%    p_t(v)       & \quad \text{if } v \in S\\
%    0  & \quad \text{otherwise} \\
 % \end{cases}
%\] 
It is clear that $p_t {\restriction_S}$ is not a probability distribution over the set $S$ as the sum could be less than $1$. 
%We define the notion of  {\em local mixing time} with respect to a source vertex as follows. 

Informally, {\em local mixing set}, with respect to a source node $s$, means that there exists  some (large-enough) subset of nodes $S$ containing $s$ such that the random walk probability distribution becomes close to the stationary distribution restricted to $S$ (as defined above) quickly. The time that a random walk mixes locally on a set $S$ is called as {\em local mixing time} which is formally defined below.    

\begin{definition}\label{def:loc-mix-time} (Local Mixing Set and Local Mixing Time) \\
Consider a vertex $s \in V$. Let $\beta \geq 1$ be a positive constant and $\epsilon \in (0,1)$ be a fixed parameter. We  first
 define the notion of local mixing in a set $S$.  Let $S  \subseteq V$ be a {\em fixed} subset containing $s$ of size at least $n/\beta$. Let $\pd_t{\restriction_S}$ be the restricted probability distribution over $S$ after $t$ steps of a random walk starting from $s$ and  $\pmb{\pi}_S$ be as defined above.  Define the {\em  mixing time  with respect to set $S$}  as  $\tau^S_{s}(\beta, \eps) = \min \{t: ||\pd_t{\restriction_S} - \pmb{\pi}_S ||_1 < \eps \}$. We say that the random walk {\em locally mixes} in $S$ if  $\tau^S_{s}(\beta, \eps)$ exists and well-defined. (Note that a walk may not locally mix in a given set $S$, i.e., there exists no time $t$ such that  $||\pd_t{\restriction_S} - \pmb{\pi}_S ||_1 < \eps$; in this case we can take $\tau^S_{s}(\beta, \eps)$ to be $\infty$.)

The local mixing time with respect to $s$ is defined as
 $\tau_{s}(\beta, \eps) = \min_{S}\tau^S_{s}(\beta, \eps)$, where the minimum is taken over all subsets $S$ (containing $s$) of size
at least $n/\beta$, where the random walk starting from $s$ locally mixes. A set $S$  where the minimum  is attained (there may be more than one) is
called the {\em local mixing set}. 
%We call $\tau_{s}(\beta, \eps)$ as {\em local mixing time} corresponding to the source vertex $s$ and the parameter $\beta$. 
The local mixing time of the graph, $\tau(\beta, \eps)$ (for  given $\beta$ and $\epsilon$), is $\max_{v \in V} \tau_{v}(\beta, \eps)$. \hfill $\blacksquare$
\end{definition}

From the above definition, it is clear that $\tau_{s}(\beta, \eps)$ always exists (and well-defined) for every fixed $\beta \geq 1 $, since in the worst-case,
it equals the mixing time of the graph; this happens when $|S| = n \geq n/\beta$ (for every $\beta \geq 1$). We note that, crucially, in the above definition of local mixing time,  the {\em minimum} is taken over subsets $S$ of size at least $n/\beta$, and thus, in many graphs, 
local mixing time can be substantially smaller than the mixing time when $\beta > 1$ (i.e., the local mixing can happen much earlier in
some set $S$ of size $\geq n/\beta$ than
the mixing time). 

%It is important to note that the set $S$ where the  local mixing time is attained is
%not fixed  a priori, it only requires that a set $S$ of size at least $n/\beta$ exists. Further, it is also clear that for a specific length of a random walk, the local mixing set may or may not exists. The {\em largest local mixing set} $S_t$ is the maximum size local mixing set for a given length $t$ of the random walk (if exists). In this paper we consider $\eps$ to be $1/2e$ which is standard assumption in the definition of mixing time \cite{Levin}. We use the term ``a random walk mixes on a set $S$'' to mean that the random walk probability distribution (at time $t$) is close to restricted stationary distribution on $S$, i.e., $||\pd_t{\restriction_S} - \pmb{\pi}_S ||_1 < 1/2e$.  

%\onlyShort{The full version of the paper can be found in: add site}

%
%\subsection{Our Contributions}\label{sec:contribution}

\vspace{-0.1in}
\section{related works}
\vspace{-0.05in}
\label{sec:related}
There has been extensive work on  community detection in graphs, see, e.g., 
the surveys of \cite{abbe2018community,abbe2015community,fortunato2010community, fortunato2016community}. 
Here we  focus mainly on related works in distributed community detection and in the SBM, especially the $G_{npq}$ model.

 Dyer and Frieze \cite{dyer1989solution} show that if $p > q$ then the minimum edge-bisection is the one that separates the two classes and present an algorithm  that gives the bisection in $O(n^3)$ expected time. Jerrum and Sorkin improved  this bound for some range of $p$ and $q$ by using simulated annealing. Further improvements and more efficient algorithms were obtained in \cite{condon2001algorithms, mossel2012stochastic}. We note that all the above algorithms are  centralized and based on expensive procedures such as simulated annealing and spectral graph computations: all of them require the full knowledge of the graph.
 
 The work of Clementi et. al \cite{clementi-tcs15} is notable because they present a distributed protocol based on the popular Label Propagation approach and prove that, when the ratio $p/q$ is larger than $n^b$ (for an arbitrarily small constant $b >0$), the protocol finds the right planted partition in $O(\log n)$ time. Note that however, they consider only two communities in their PPM model. We also note that this ratio can be significantly weaker compared to the ratio  (for identifying all the $r$ communities) derived in our Theorem \ref{thm:r-community}
 which is $p/q = O(r\log(n/r)$ where $r$ is the number of communities (which can be much smaller compared to $n$, the
 total number of vertices). Also our algorithm works for any number of communities.
 
Random walks have been successfully used for graph processing in many ways. Community detection algorithms use the statistics of random walks to infer structure of graphs as clusters or communities. 
{\em Netwalk} \cite{zhou2004network} defined a proximity measure between neighboring vertices of a graph. Initialized each node as a community, it merges two communities with the lowest proximity index into a community in iterations. But it is an expensive method running in $O(n^3)$ time complexity. 
Similarly, {\em Walktrap} by Pons et al \cite{pons2006computing} defined a distance measure between nodes using random walks. Then they merged similar nodes into a community in iterations. The logic behind their method is that random walks get trapped inside densely connected parts of a graph which can capture communities. Their algorithm is centralized and has expensive run-time of $O(mn^2)$ in worst case. 

Some works uses linear dynamics of graphs to perform basic network processing tasks such as reaching self-stabilizing consensus in faulty distributed systems \cite{benezit2009interval, olshevsky2011convergence} or Spectral Partitioning \cite{donath1972algorithms, lei2015consistency,  peng2015partitioning}. They work on connected non-bipartite graphs. Becchetti et al \cite{becchetti2017find} defines averaging dynamics, in which each node updates its value to the average of its neighbors, in iterations. It partitions a graph into two clusters using the sign of last updates.  Another interesting research by Becchetti et. al. \cite{DBLP:conf/esa/BecchettiCMNPRT18} used random walk to average values of two nodes when randomly any two nodes meet and showed that it ends in detecting communities.  
The convergence time of the averaging dynamics on a graph is the mixing time of a random walk \cite{shah2009gossip}. These methods works well on graphs with good expansion \cite{hoory2006expander} and are slower on sparse cut graphs.
{\em Label Propagation Algorithm (LPA)} \cite{raghavan2007near} is another updating method which converges to detecting communities by applying majority rule. Each node initially belongs to its own community. At each iteration, each node joins a community having majority among its neighbors, applying a tie-breaking policy. Recently Kothapalli et al provided a theoretical analysis for its behavior \cite{kothapalli2013analysis} on {\em dense} PPM graphs ($p = \Omega(1/n^{1/4})$ and $q = O(p^2)$). In comparison, our algorithm works even for the more challenging case of sparse graphs ($p = \Omega(\log n/n)$, i.e., the near the connectivity threshold).   A major drawback of LPA algorithm is the lack of convergence guarantee. For example, it can run forever on a bipartite graph where each part gets a different label (each community is specified by a label).

\iffalse
Some other random walk methods use seed node expansion to detect community structures. Starting from a small set of ``seed" nodes (e.g., a single node or a connected cluster of nodes), they grow a community around the seed set by using a goodness function. They may use a greedy method \cite{mislove2010you} or an optimization function \cite{bagrow2008evaluating} to improve. Optimization techniques are the natural way to search for clusters and communities of graphs. In this method, a quality function  (that captures the ``goodness" of clusters) is defined and then its extreme for an optimal output is searched. For example, Newman et al uses modularity as their optimization function \cite{newman2004finding}. Their aim is to maximize the modularity function that  results  in detecting stronger community structures.
The method in \cite{clauset2005finding} grows the size of community by adding nodes which improve the modularity most; it stops when reaching a favored community size. Recently Hollocou et al \cite{hollocou2018multiple} applied  clustering methods to find seed nodes and then find communities from those nodes. Some other methods \cite{andersen2006communities, kloumann2014community, whang2013overlapping} use a small seed set instead of one node for their detection. Unlike those methods, our CDRW algorithm uses local mixing property of a graph \cite{sarma2012fast, molla2017distributed,MP18} to detect community structure around a seed node.
\fi

 Although this paper uses the notion of local mixing time introduced in \cite{MP18}, there are substantial differences. In \cite{MP18}, the authors consider only the local mixing time which is essentially the existence of a mixing set of certain size, but {\em not the set of nodes} where the random walk mixes. The computation of the local mixing set is more challenging. A key idea of this paper is to use this notion to identify communities. For this, the algorithm and approach of \cite{MP18} has to be modified substantially.

\section{Algorithm for Community Detection}\label{sec:algo}
%\vspace{-0.05in}
We design a random walk based community detection algorithm (cf. Algorithm~\ref{alg:via-local-mixing}). Given a graph and a node, the algorithm finds a community containing the node in a distributed fashion. We show the efficiency and effectiveness of the algorithm  both theoretically and experimentally on random graphs and stochastic block model. 

%We present a distributed algorithm based on random walks. We theoretically prove the effectiveness of this algorithm on Erd\H{o}s-R{\'e}ny random graphs and planted partition models. Indeed we prove that our algorithm detects a random graph as a single community. We then prove that it detects $k$ ground truth communities of $k$ partition $Planted\ Partition\ Model$ graphs. We show it beats or work the same of the comparing baseline methods on different synthesis graphs. In this work we focused on the theoretical proof and experimental results on synthesis graphs.

%\subsection{Community Detection via Local Mixing Set}\label{sec:local-mixing-algo}
\noindent \textbf{Outline of the Algorithm.}  We use the concept of {\em local mixing set}, introduced by Molla and Pandurangan \cite{MP18}, to identify community in a graph. A local mixing set of a random walk is a subset of the vertex set where the random walk probability mixes fast, see formal definition in the Section~\ref{sec:model}. Intuitively, a random walk probability mixes fast over a subset where nodes are well-connected among themselves. The idea is to use the concept of local mixing set to identify a community --- a subset where nodes are well-connected  inside the  set and less-connected outside. That is, if a random walk starts from a node inside a community, its probability distribution is likely to mix fast inside the community nodes and  less amount of probability will go outside of the set. Thus the high level idea of our approach is to perform a random walk from a given source node and find the largest subset (of nodes) where the random walk probability mixes quickly. We extend the distributed algorithm from \cite{MP18} to find a largest mixing set in the following way. In each step of the random walk, we keep track the size of the largest mixing set. When the size of the largest mixing set is not increasing significantly with the increase of the length of the random walk, we stop and output the largest mixing set as the community containing the source node.

\iffalse
Given a graph, our algorithm randomly picks a source node and computes the community containing the source node using the above approach. To find other communities, we can recurse the algorithm from different source nodes belonging to different communities. After finding a community, the algorithm picks another source node from the remaining nodes (i.e., outside of the computed community nodes) and recurses with the new source node to find a different community. 
It is not difficult to adapt this algorithm to find communities in parallel by starting random walks from several nodes (e.g.,
these can be chosen randomly). \anis{need to modify this part depending on this parallel community computation} 
%We don't further address this aspect in this paper,  as our main focus is on
%analyzing the effectiveness of the algorithm in finding communities.

% The detail algorithm is discussed below. We prove the effectiveness and efficiency of our algorithm by both theoretical analysis and comprehensive experimental evaluation.     
\fi

\noindent \textbf{Algorithm in Detail.} %\label{sec:local-mixing-algo}
Given an undirected graph $G(V,E)$, our algorithm %detects communities one by one. The algorithm 
randomly selects a node $s$ and outputs a community set $C^s \subseteq V$ containing $s$. It maintains a set, called as $pool$ which contains all the remaining nodes of $V$ excluding the nodes in $C^s$. Then another random node gets  selected from the set $pool$ and we  compute the community containing that node in $G$, and so on. This way all the different communities are computed one by one. The $pool$ set is initialized by $V$ in the beginning. The algorithm  stops when the $pool$ set becomes empty. 

Now we describe how the algorithm computes the community set $C^s$ in $G$ from a given node $s$. The algorithm performs a random walk from the source node $s$ and computes the probability distribution $\pd_{\ell}$ at each step $\ell$ of the random walk. The probability distribution $\pd_{\ell}$ starting from the source node $s$ is computed locally by each node as follows: Initially, at round $0$, the probability distribution is: at node $s$, $p_0(s, s) = 1$ and at all other nodes $u$, $p_0(s, u) = 0$.  At the start of a round $\ell$, each node $u$ sends $p_{\ell-1}(s, u)/d(u)$ to its $d(u)$-neighbors and at the end of the round $\ell$, each node $u$ computes $p_{\ell}(s, u) = \sum_{v \in N(u)} p_{\ell-1}(s, v)/d(v)$. This local flooding approach essentially simulates the probability distribution of each step of the random walk starting from a source node.  Moreover, this deterministic flooding approach can be used to compute the probability distribution $\pd_{\ell}$ of length $\ell$ from the previous distribution $\pd_{\ell-1}$ in {\em one round} only-- simply by resuming the flooding from the last step. The full algorithm can be found in Algorithm~1 in \cite{MP18}. Then at each step $\ell$, our algorithm computes a largest mixing set $S_{\ell}$. The largest  mixing set $S_{\ell}$ is computed as follows: Each node $u$ knows its $p(u) = p_{\ell}(s, u)$ value. The algorithm gradually increases the size of a candidate local mixing set $S$ starting from size $1$.\footnote{In the pseudocode we assume the size of each community is at least $\log n$.} First each node $u$ locally calculates its $x_u$ value as $x_u=|p_{\ell}(u)-\frac{d(u)}{\mu'(S)}|$, where $\mu'(S) = \frac{2m}{n}|S|$ is the average volume of the set $S$. Note that any node $u$ can compute $\mu'(S)$ when it knows the ``size''-- $|S|$ and hence can compute $x_u$ locally. However, it's difficult to compute $\mu(S)$ unless it knows the set $S$ (i.e., the nodes in $S$) and the degree distribution of the nodes in $S$. Computing nodes in $S$ and their degree distribution is expensive in terms of time. That's why we  consider $\mu'(S)$ instead $\mu(S)$ in the localized algorithm. 

Then the source node $s$ collects $|S|$ smallest of $x_u$ values and checks if their sum is less than $1/2e$ (mixing condition). For this each node may sends its $x_u$ to the source nodes $s$ via upcasting through a BFS tree rooted at $s$. (A BFS tree is computed from $s$ at the beginning of the algorithm). However, the upcast may take $\Omega(n)$ time in the worst case due to the congestion in the BFS tree. A better approach is used in \cite{MP18}, which is to do a binary search on $\{x_u \,|\, u\in V\}$, as follows: All the nodes send $x_{\min}$ and $x_{\max}$ (the minimum and maximum respectively among all $x_u$) to the root $s$ through a convergecast process (e.g., see \cite{peleg}). This will take time proportional to the depth of the BFS tree. Then $s$ can count the number of nodes whose $x_u$ value is less than $x_{\med} = (x_{\min} + x_{\max})/2$ via a couple of broadcast and convergecast. In fact, $s$ broadcasts the value $x_{\med}$ to all the nodes via the BFS tree and then the nodes whose $x_u$ value is less than $x_{\med}$ (say, the {\em qualified nodes}), reply back with $1$ value through the convergecast. Depending on whether the number of qualified nodes is less than or greater than $|S|$, the root updates the $x_{\med}$ value (by again collecting $x_{\min}$ or $x_{\max}$ in the reduced set) and iterates the process until the count is exactly $|S|$. Note that there might be multiple nodes with the same $x_u$ value. We can make them distinct by adding a `very' small random number to each of the $x_u$ such that the addition doesn't affect the mixing condition. The detailed approach and analysis can be found in \cite{MP18}.   

Once the node $s$ gets $|S|$ smallest $x_u$s, it checks if their sum is less than $1/2e$. If true, then these nodes $u$ whose sum value is less than $1/2e$ gives a candidate mixing set and its size is $|S|$. Then we increase the set size and check if there is a larger mixing set. If the mixing condition doesn't satisfy, then there is no mixing set of size $|S|$. The algorithm iterates the checking process few more times by increasing the size of $S$ and checking if there is a mixing set of larger size. If not, then the algorithm stops for this length $\ell$ and stores the largest mixing set at $s$. This way, the algorithm finds the largest mixing set $S_{\ell}$ at the $\ell^{th}$ step of the random walk. Note that we can increase the candidate mixing set size by $1$ each time. This will increase the time complexity of the algorithm by a factor of the ``size of the largest mixing set''. Instead we increase the size of the mixing set by a factor of $(1+1/8e)$ in each iteration. This will only add a factor of $O(\log n)$ to the time complexity. The reason why we increase by a factor of $(1+1/8e)$ instead of doubling is discussed in \cite{MP18} (see, Lemma~3 in \cite{MP18}). The correctness of the all the above tests is also analyzed in \cite{MP18}. 

Then the algorithm checks if the size of the largest mixing set $S_{\ell}$ at step $\ell$ increases significantly than mixing set $S_{\ell-1}$ in the previous step $(\ell -1)$. This is checked locally by the source node as the source node has the information of the largest mixing set of the current and previous steps. If the size doesn't increase by a factor $(1+\delta)$, i.e., if $S_{\ell} < (1+\delta)S_{\ell-1}$, then the algorithm stops and output $S_{\ell-1}$ as the community set $C^s$. Otherwise, the algorithm increases the length by $1$ and checks for $S_{\ell +1}$. The parameter $\delta$ is chosen to be the conductance of the graph  $\Phi_G$ which essentially measures the vertex expansion of the graph.   

%%%BEGIN ALGO BOX
\setlength{\textfloatsep}{4pt}
\begin{algorithm*}[h]
\footnotesize
	\caption{\sc Community-Detection-by-Random-Walks (CDRW)}
	\label{alg:via-local-mixing}
\textbf{Input:} An undirected graph $G = (V, E)$. \\
\textbf{Output:} Set of Detected Communities $\mathcal{C^D}$.
	\begin{algorithmic}[1]
		\State $\mathcal{C^D} \gets \{\}$
		\State $pool \gets V$
		\While{$pool\neq \emptyset$} \Comment{There exist nodes not assigned to any communities yet}
		\State $s \gets$ pick a random node from $pool$
		\State $s$ computes a BFS tree  of depth $O(\log n)$ via flooding 
		\State Set $R = \log n$. \Comment{Size of the local mixing set initialize by $\log n$. Assume community size is $\geq \log n$}
		\State Set $p_0(s) = 1$ and $p_0(u) = 0$ for all other nodes $u$. 
%		\State Set $R = \log n$, $p_0(s) = 1$, and $p_0(u) = 0$ for all other nodes $u$.
		\For{$\ell=1, 2, 3, \dots, O(\log n)$} \Comment{Length of the random walk}
		\State Each node $u$ whose $p_{\ell-1}(u) \neq 0$, does the following in parallel:\\
 			\hspace{2cm} (i) Send $p_{\ell-1}(u)/d(u)$ to all the neighbors $v \in N(u)$. \\
 			\hspace{2cm} (ii) Compute the sum of the received values from its neighbors and sets it as $p_{\ell}(u)$.  
		\For{$|S| = R, (1 +1/8e)R, (1 +1/8e)^2 R, \ldots, n$}
		\State Each node $u$ computes the difference $x_u= |p_{\ell}(u)-\frac{d(u)}{\frac{2m}{n}|S|}|$ locally  
		\State $s$ computes the sum of $|S|$ smallest $x_u$ values using binary search method discussed in the detail description of the algorithm.
		\State  $s$ checks if the sum is less than $1/2e$, i.e., if $\sum_{\{|S|\, smallest \, x_u\}} x_u<\frac{1}{2e}$.
		\State If ``true", then $s$ checks for the next size of the mixing set  
		\State Else, $s$ sets $S_{\ell}$ to be the largest set $S$ which satisfies the mixing condition. $s$ broadcasts an indicator message to all the nodes via BFS tree. The nodes whose $x_u$ value gives the $|S|$ smallest values belong to the largest mixing set $S_{\ell}$.  
		\EndFor
%		\State $s$ checks the community condition as follows:
%		\State \textbf{if} $\frac{|S_{\ell}|}{|S_{\ell-1}|} <(1+\delta)$ \textbf{Then} Break the {\bf for}-loop. \Comment{$\delta =  \Phi_G$}
		
%		\State $s$ checks the community condition: \textbf{if} $\frac{|S_{\ell}|}{|S_{\ell-1}|} <(1+\delta)$ \textbf{Then} Break the {\bf for}-loop. \Comment{$\delta =  \Phi_G$}
		
		\If{$\frac{|S_{\ell}|}{|S_{\ell-1}|} <(1+\delta)$} \Comment{$\delta =  \Phi_G$}
		\State Break the {\bf for}-loop. 
		\EndIf

		\EndFor
		
		\State $C^s \gets S_{\ell-1}$ 
		\State $\mathcal{C^D} \gets \mathcal{C^D} \cup \{C^s\}$
		\State $pool \gets pool \setminus C^s$
		
%		\State $C^s \gets S_{\ell-1}$; $\mathcal{C^D} \gets \mathcal{C^D} \cup \{C^s\}$; $pool \gets pool \setminus C^s$; %\Comment{Include $C^s$ \hspace*{5mm}in $\mathcal{C}$ and remove the community nodes from the $pool$}
		\EndWhile
		\State Return $\mathcal{C^D}$ %\Comment{set of communities}
	\end{algorithmic}
\end{algorithm*}

%%END ALGOBOX

\subsection{Analysis}\label{sec:local-mixing-analysis}
\vspace{-0.05in}
%\noindent \textbf{Analysis.} 
We analyze the algorithm and show that it correctly identifies communities in the planted partition model (PPM) -- $G_{npq}$ graphs. The $G_{npq}$ graph is formed by connecting several communities of $G_{np}$ graphs (see the definition in Section~\ref{sec:model}). 
Let us first analyze the Algorithm~\ref{alg:via-local-mixing} on the random  graph $G_{np}$. We then extend the analysis to the stochastic block model $G_{npq}$.
%We show that the algorithm stops in at most $O(\log{n})$ times and returns largest locally mixing set which resembles local community. In our analysis, distances are hop distance between nodes. We analytically prove that it works on $random$ and $planted\ random$ graphs. 

\noindent \textbf{On $G_{np}$ Graphs.} 
Suppose the algorithm is executed on the standard random (almost regular) graph $G_{np} = (V_1, E_1)$, defined in Section~\ref{sec:model}. Since $G_{np}$ is an expander graph, the random walk starting from any node mixes over the vertex set $V_1$ very fast; in fact in $O(\log n)$ steps. Given any node $s$, we show that our algorithm computes the community $C^s$ as the complete vertex set $V_1$. More precisely, we show that the size of the largest mixing set increases on a higher rate (than the considered threshold) after each step of the random walk, when the length of the walk is $o(\log n)$.  Since $O(\log n)$ is the mixing time of $G_{np}$, the random walk probability reaches  the stationary distribution after $c\log n$ steps, for a sufficiently large constant $c$. 

Let $p_t(u)$ be the probability that the walk is at $u$ after $t$ steps (starting from a source node $s$). It is known that in a regular graph $p_t(u)$ is bounded by: 
\begin{equation}\label{eq1}
\frac{1}{n} - \lambda_2^t  \leq p_{t}(u) \leq \frac{1}{n} + \lambda_2^t 
\end{equation}
where $\lambda_2$ is the second largest eigenvalue (absolute value) of the transition matrix of $G_{np}$\footnote{The bound follows from the standard bound $|p_{t}(s, u) - \pi(v)| \leq \lambda^t_2\sqrt{\pi(v)/\pi(s)}$ in general graphs \cite{LL93}. In a regular graph, $\pi(v) = 1/n$ for all $v$. Note that $G_{np}$ is not exactly a regular graph, but very close to regular (especially if $p= (c\log n)/n$ for a large enough constant $c$). It can be shown that $\pi(v) = 1/n \pm o(1/n)$ in $G_{np}$. For simplicity we assume that $G_{np}$ is a regular graph as this little $\pm \eps$  changes in the degree or in the probability distribution doesn't affect the lemmas.}. Hence, the above bound on the probability distribution $p_{t}$ holds in $G_{np}$ graphs. It is further  known that in a random $d$-regular graph, the second largest eigenvalue is bounded by \cite{Friedman-book}: 
\begin{equation}\label{eq2}
\frac{1}{\sqrt{d}} \leq \lambda_2 \leq \frac{1}{\sqrt{d}}+o(1) 
\end{equation}
Let $B_{\ell}$ be the set of nodes that are within the distance $\ell$  from the $s$. The  distance is measured by the hop distance between nodes. Let's call $B_{\ell}$ a ball of radius $\ell$ centered at $s$. We now show that after $\ell$ steps of the random walk, the largest mixing set is $B_{\ell/2}$ in a $G_{np}$ graph.
\begin{lemma}\label{lem:mixing-half-ball}
	Let a random walk start from a source node $s$ in a $G_{np}$ graph. Then for any length $\ell$ which is less than the mixing time of the random walk, the largest mixing set is the ball $B_{\floor{\ell/2}}$ with high probability. 
\end{lemma}
\begin{proof}
Assume $\ell = o(\log n)$, since $\ell$ is less than the mixing time $O(\log n)$.  It is known that the size of the ball $B_{\ell}$ in a random graph $G_{np}$ is bounded by  $ O((np)^{\ell})$ with high probability (cf. Lemma~2 in \cite{CL01}). %Therefore, $|B_{\ell}| \leq O((\log n)^{\ell})$.
To prove the lemma we show that the random walk probability mixes inside the ball $B_{\floor{\ell/2}}$ and doesn't mix on the ball of radius larger than $\floor{\ell/2}$. Recall that the condition of locally mixing on a subset $B_{\floor{\ell/2}}$ is $\sum_{u\in B_{\floor{\ell/2}}} {|p_{\ell}(u)-\frac{1}{|B_{\floor{\ell/2}}|}|}\leq \frac{1}{2e}$, (since $G_{np}$ is regular graph). Using the above bound of $p_t(u)$, $\lambda_2$ (Equ~\ref{eq1}, \ref{eq2}) and $|B_{\floor{\ell/2}}| \leq d^{{\floor{\ell/2}}}$ (since $d = np = \Theta(\log n)$ in expectation in $G_{np}$) we have:
\begin{align*} \label{eq2}
& \sum_{u\in B_{\floor{\ell/2}}} {\Big|p_{\ell}(u)-\frac{1}{|B_{\floor{\ell/2}}|}\Big|}  
 \leq \sum_{u\in B_{\floor{\ell/2}}} {\Big|\frac{1}{n}+ \lambda_2^{\ell}-\frac{1}{|B_{\floor{\ell/2}}|}\Big|} \\
%& = \sum_{u\in B_{\floor{\ell/2}}} {\Big|\frac{1}{n}+\frac{1}{d^{\frac{\ell}{2}}} + o(1) -\frac{1}{|B_{\floor{\ell/2}}|}\Big|} \\ 
& \leq  \, d^{{\floor{\ell/2}}} \Big|\frac{1}{n}+\frac{1}{d^{\frac{\ell}{2}}} + o(1) -\frac{1}{d^{{\floor{\ell/2}}}}\Big|  < \, \frac{d^{{\floor{\ell/2}}}}{n} + o(1)\\
  &  < \frac{1}{2e} \hspace{1cm} [\text{since $\frac{d^{{\floor{\ell/2}}}}{n} = o(1)$ as $\ell = o(\log n)$}]
\end{align*} 
This shows that the random walk of length $\ell$ mixes over the nodes in $B_{\floor{\ell/2}}$. Now we show that it doesn't mix on $B_{t}$ for $t >\ell/2$. Again from Equation~\ref{eq1} and~\ref{eq2}, 
\begin{align*}
& \sum_{u\in B_{t}} {\Big|p_{\ell}(u)-\frac{1}{|B_{t}|}\Big|} 
 \geq \sum_{u\in B_{t}} {\Big|\frac{1}{n} - \lambda_2^{\ell}-\frac{1}{|B_{t}|}\Big|} \\
 & = \sum_{u\in B_{t}} {\Big|\lambda_2^{\ell} + \frac{1}{|B_{t}|} - \frac{1}{n} \Big|} 
 \geq \sum_{u\in B_{t}} {\Big|\lambda_2^{\ell} \Big|} \hspace{.5cm} [\text{since $|B_t| \leq n$}] \\
& \geq \frac{|B_{t}|}{d^{\frac{\ell}{2}}} \geq  d > 1/2e    \hspace{.5cm} [\text{since $t >\ell/2$ and $d=\log n$}] 
\end{align*} 
Thus the largest mixing set is $B_{\floor{\ell/2}}$. 
\end{proof} 

Now we show that our algorithm outputs the full vertex set as the community in $G_{np}$ graphs.  

\begin{lemma}\label{lem:gnp-community}
	Given a random regular expander graph $G_{np} = (V_1, E_1)$, the Algorithm~\ref{alg:via-local-mixing} outputs the vertex set $V_1$ as a single community with high probability.  
\end{lemma}
	\begin{proof}
		It follows from the previous lemma that when $\ell$ is less than the mixing time of $G_{np}$, then the largest local mixing set is  $B_{\floor{\ell/2}}$. Therefore, in each step of the random walk, the size of the mixing set is increased by a factor $\frac{|B_{\floor{\ell/2}}|}{|B_{\floor{\ell/2-1}}|} = O(d) = \Theta(\log n) >(1+\delta)$. Hence, by the condition of the Algorithm~\ref{alg:via-local-mixing}, it doesn't stop and continue to look for a community set for the larger lengths of the random walk. It means, until the length of the random walk reaches to the mixing time of the graph $G_{np}$, the algorithm continue its execution. When the length reaches the mixing time, then the random walk will mix the full vertex set $V_1$. Then the algorithm stops and outputs $V_1$ as a single community set (as the size of the mixing set won't increase anymore for larger lengths).      
	\end{proof}
	
\noindent \textbf{On $G_{npq}$ Graphs.}  
Let us now analyze the algorithm on the planted partition model i.e., on a random $G_{npq}$ graph. A random $G_{npq}$ graph is formed by $r$ equal size blocks $C_1,C_2,\dots, C_r$ where each component $C_i$ is a  $G_{\frac{n}{r}p}$ random graph (see the definition in Section~\ref{sec:model}). We show that the algorithm correctly identifies each block as a community. Suppose the randomly selected node $s$ belongs to some block $C$. The induced subgraph on $C$ is a $G_{\frac{n}{r}p}$ graph i.e., the nodes inside $C$ are connected to each other with probability $p$. Further each node in $C$ is connected to every node outside of $C$ with probability $q$. Thus the random walk may go out of the set $C$ at some point. We show that the probability of going out of $C$ is very small when the length of the walk is smaller than the mixing time of $G_{\frac{n}{r}p}$ graph, which is $O(\log (n/r))$. 

\begin{lemma}\label{lem:rw-stays-inside}
Given a $G_{npq}$ graph and a node $s$ in some block $C$, the probability that a random walk starting from $s$ stays inside $C$ is at least $1-o(1)$ until $\ell = O(\log (n/r))$ when $q = o(\frac{p}{r\log (n/r)})$. 
\end{lemma}
\begin{proof}  
We show that in each step, the probability that the random walk goes outside of $C$ is $o(1/\log n)$. For any $u\in C$, the number of neighbors of $u$ in $C$ is $p|C| = pn/r$ and the number of neighbors in $\bar{C}= V\setminus C$ is $q|\bar{C}| = q(n-n/r)$ in expectation. Thus the probability that the random walk goes outside of the block $C$ is $\frac{q(n-n/r)}{p(n/r) + q(n-n/r)} = \frac{q(r-1)}{p + q(r-1)}$. This is $o(1/\log (n/r))$ when $q = o(\frac{p}{r\log (n/r)})$. Thus in $\ell = O(\log (n/r))$ steps, the probability that  walk goes outside of the block $C$ is $o(1)$. That is the random walk stays inside $C$ with probability at least $1-o(1)$. 
\end{proof}

Now we show that the random walk probability will mix over $C$ in $O(\log (n/r))$ steps. 

\begin{lemma}\label{lem:inside-gnp-mixing}
	Given a $G_{npq}$ graph and a node $s \in C$, a random walk starting from $s$ will mix over the nodes in $C$ after $\tau = O(\log (n/r))$ steps with high probability.  
\end{lemma}
\begin{proof}   
We show that after $O(\log (n/r))$ steps of the walk, the amount of probability goes out of $C$ is  very little and that the remaining probability will mix inside $C$. The expected number of outgoing edges from any subset $S$ of the block $C$ is $|E(S, V\setminus C)| = q|\bar{C}||S| = q(n-n/r)(|S|)$. In each step the amount of probability goes out of $C$ is $\frac{|E(S, V\setminus C)|}{d|S|}$, as $d= p(n/r) + q(n-n/r)$ is the degree of a node, each edge carries $1/d|S|$ fraction of the probability. We have $\frac{|E(S, V\setminus C)|}{d|S|} = \frac{q(n-n/r) |S|}{(p(n/r) + q(n-n/r))|S|} = o(1/\log (n/r))$ for $q = o(\frac{p}{r\log (n/r)})$. Thus in $\ell = O(\log (n/r))$ steps, the amount of probability goes out of $C$ is $o(1)$. Hence $1-o(1)$ fraction of the probability remains inside $C$ and it will mix over the nodes in $C$ after $O(\log (n/r))$ steps as shown in the above Lemma~\ref{lem:mixing-half-ball} and~\ref{lem:gnp-community}. 
\end{proof}

Thus it follows from the above lemma that the largest mixing set is $C$ after $\tau = O(\log (n/r))$ steps of the random walk. Further, it is shown in Lemma~4 of \cite{MP18} that the random walk keeps mixing in $C$ until $2\tau$ steps. In other words, $C$ remains the largest local mixing set for at least another $\tau$ steps. Thus the size of the largest local mixing set will not increase from $C$ in the further few steps of the walk after the mixing time $\tau$. Hence the algorithm outputs $C$ as a community with high probability. Since we sample the source node $s$ from the different blocks, each time our algorithm outputs a new community until all the blocks are identified as separate communities.
	
The $\delta$ value measures the rate of change of the size of the largest mixing set in each step. When the largest mixing set reaches a community $C$, the vertex expansion becomes $\frac{|E(C, V\setminus C)|}{d|C|}$ which is the conductance of the $G_{npq}$ graph. If the largest mixing set doesn't reach  the community, the size increases in higher rate than $\delta$. Hence we take $\delta$ to be $\Phi_G$ in our algorithm to stop and output the community. We assume that $\Phi_G$ is given as input, or it can be computed using a distributed algorithm, e.g., \cite{Kuhn2015DistributedSC}.

%% Running Time Discussion %%%%
\noindent \textbf{Complexity of the Algorithm in the CONGEST model.}
Let us first analyze the {\em distributed time complexity} of the Algorithm~\ref{alg:via-local-mixing} which computes a community corresponding to a given source node. We will focus on the CONGEST model first. The algorithm first computes a BFS tree of depth $O(\log n)$ from the source node. This takes $O(\log n)$ rounds. Note that the diameter of a $G_{np}$ graph is $O(\log n)$; hence the BFS tree covers all the nodes in the community containing the source node. The algorithm then iterates for the length of the walk, $\ell = 1, 2, 4, \ldots, O(\log n)$. In each iteration: 
\begin{itemize}
\item The algorithm probability distribution $\pd_{\ell}$. As we discussed before, it takes $O(1)$ rounds to compute $\pd_{\ell}$ from $\pd_{\ell-1}$. 
\item $s$ collects the sum of $|S|$ smallest $x_u$s through the BFS tree using binary search method. It takes $O((depth\_BFS\_tree) \cdot \log n) = O(\log^2 n)$ rounds. This is done for all the potential candidate set of size $(1+1/8e)^i|S|$, where $i=0, 1, 2, \ldots$. It may take $O(\log n)$ rounds in the worst case. Hence the total time taken is $O(\log^3 n)$ rounds.
\item  Checking if the sum of differences is less than $1/2e$ and also checking the community condition is done locally at $s$.
\end{itemize} 
Thus the total time required is $O(\log n) + O(\log n) \cdot (O(1) + O(\log^3 n))$, which is bounded by $O(\log^4 n)$. 

%It is shown in \cite{MP18} that local mixing time can be computed in $O(\tau^{loc})\log^3 n$ time in a regular graph, where $\tau^{loc}$ is the local mixing time of the graph corresponding to a source node $s$. The local mixing time algorithm in \cite{MP18} uses several binary search (such as on the length of the random walk, on the size of the mixing set, to find $|S|$ smallest $x_u$ values) to achieve the time complexity. Their algorithm essentially computes local mixing set by checking if there is a mixing set of certain size in each step of the random walk. Using the similar ideas and analysis, our algorithm computes the community $V_1$ in $O(\log^4 n)$ time in $G_{np}$, since $\tau^{loc}$ becomes the mixing time of $G_{np}$ which is $O(\log n)$. 

\noindent \textbf{Message Complexity of the Algorithm.}
Let us calculate the number of messages used by the algorithm during the execution in a $G_{npq}$ graph. The degree of a node is $p(n/r) + q(n-n/r)$ in expectation. Hence the number of edges in the $G_{npq}$ graph is $n^2p/r + nq(n-n/r)$.
In the worst case, the the algorithm runs over all the edges in the graph. Thus the message complexity of the algorithm for computing a single community is bounded by (time complexity) $\times$ (the number of edges involved during the execution), which would be $O(\frac{n^2}{r}(p+q(r-1)) \log ^4 n)$ in expectation. That is the message complexity of the Algorithm~\ref{alg:via-local-mixing} is $\tilde{O}(\frac{n^2}{r}(p+q(r-1)))$.

Therefore we have the following main result. 
\begin{theorem}\label{thm:main} 
Consider a stochastic block model $G_{npq}$ with $r$ blocks, where $p = \Omega(\frac{\log n}{n})$ and $q = o(\frac{p}{r\log (n/r)})$. Given a node $s$ in the $G_{npq}$ graph, there is a distributed algorithm (cf. Algorithm~\ref{alg:via-local-mixing}) that computes the block containing $s$ as a community with high probability in $O(\log^4 n)$ rounds and incurs $\tilde{O}(\frac{n^2}{r}(p+q(r-1)))$ messages in expectation.     
\end{theorem}

The CDRW algorithm can be used to detect all the $r$ communities in the PPM graphs one by one. In that case the running time would be $r$ times the time of detecting one community, which is $O(r \log^4 n)$. The message complexity in this case would be $O(n^2(p+q(r-1))\log^4 n)$ in expectation. Thus we have the following theorem.

\begin{theorem}\label{thm:r-community} 
Given a stochastic block model $G_{npq}$ with $r$ blocks, where $p = \Omega(\frac{\log n}{n})$ and $q = o(\frac{p}{r\log (n/r)})$, there is a distributed algorithm (cf. Algorithm~\ref{alg:via-local-mixing}) that correctly computes each block as a community with high probability and outputs all the $r$ communities in $O(r\log^4 n)$ rounds and incurs expected $\tilde{O}(n^2(p+q(r-1)))$ messages.  
\end{theorem}

\vspace{-0.1in}
\subsection{Complexity in the $k$-machine model.}
\vspace{-0.05in}
\label{subsec:kmachine}

As mentioned earlier, in the $k$-machine model, the input (SBM) graph is partitioned across the $k$ machines according
to the random vertex partition (RVP) model (cf. Section \ref{sec:model}).
The algorithm can be implemented in the $k$-machine model by simulating the corresponding CONGEST model algorithm.
Note that since each vertex and its incident edges are assigned to a machine (i.e., its ``home" machine --- cf. Section \ref{sec:model}),
the machine simply simulates the code executed by the vertex in the CONGEST model. If a vertex $u$ sends a message to its neighbor $u$ in the CONGEST model, then the home machine of $u$ sends the same message to the home machine of $v$
(addressing it to $v$). If $u$ and $v$ have the same (home) machine, then no communication is incurred, otherwise there will
be communication along the link that connects these two home machines.  This type of simulation is
detailed in \cite{KlauckNPR15}.  Hence one can use the Conversion Theorem (part a) of \cite{KlauckNPR15}  to compute the 
round complexity of the CDRW implementation in the $k$-machine model which depends on the message complexity and time complexity of CDRW in the CONGEST model. If $M$ and $T$ are the message and time complexities (respectively) in the 
CONGEST model, then in the $k$-machine model then by the Conversion Theorem, the above simulation will give
a round complexity of $\tilde{O}(M/k^2 + (\Delta T)/k)$, where $\Delta$ is the maximum degree of the graph.\footnote{$\tilde{O}$ notation hides a $\mbox{polylog } n$ multiplicative and additive factor.} For the SBM model,
$\Delta = O(np/r + (n-n/r)q)$. Hence plugging in the message complexity and time complexity from the CONGEST model analysis, we have that the round complexity in the $k$-machine model is $\tilde{O}( (\frac{n^2}{k^2} + \frac{n}{kr}) (p+q(r-1)) )$.

\section{Experimental Results}
In this section we experimentally analyze the performance of our algorithm in the PPM model under various parameters. In particular, we show how accurately our algorithm can identify the communities  in the PPM model.
As an important special case, we also analyze the case when $r=1$, i.e., there is only one community --- in other words, the whole graph is a $G(n,p)$ random graph. In this case, we expect the algorithm to output the whole graph as one community. 

\par
Since in the PPM model, we know the ground-truth communities, we use F-score metric \cite{info-book} to measure the accuracy of the detected communities. Let $C^D$ be the set of detected communities by CDRW algorithm and $C^G=\cup C_i$ be the ground-truth communities (each $C_i$ is a ground-truth community). Let $C^s$ be the detected community by CDRW using seed node $s$ and $C^g$ be the ground-truth community that seed node $s$ belongs to. Then the {\em precision} is the percentage of truly detected members in detected community defined as $precision(C^s)=\frac{|C^s \cap C^g|}{|C^s|}$ and {\em recall}   is the percentage of truly detected members from the ground truth community defined as $recall(C^s) = \frac{|C^s \cap C^g|}{|C^g|}$. Both {\em precision} and {\em recall} return a high value when a method detects communities well. For example, if all the detected members belong to the ground-truth community of the seed node, then its {\em precision} is equal to $1.0$; and if all the ground-truth community members of the seed node are included in the detected community, then its {\em recall} value is equal to $1.0$.
We utilize F-score as our accuracy measurement metric which reflects both precision and recall of a result. F-score
of a detected community $C^s$ is defined as:\small
%\begin{equation}
%\label{eq:fscorei}
$F\text{-}score(C^s) = \frac{2 \times precision(C^s) \times recall(C^s)}{precision(C^s)+recall(C^s)}$ \normalsize.
%\end{equation}
Then the total {\em F-score} is equal to the average F-score of all detected communities: \small $F\text{-}score = \frac{1}{|C^D|} \sum_{C^j \in C^D} {F\text{-}score(C^j)}$ \normalsize.
%\begin{equation}
%\label{eq:fscore}
%F\text{-}score = \frac{1}{|C^D|} \sum_{C^j \in C^D} {F\text{-}score(C^j)}
%\end{equation}
Again a higher F-score value means a better detection of communities.
 \begin{figure}[t]
 	\centering
	\includegraphics[width=0.47\textwidth]{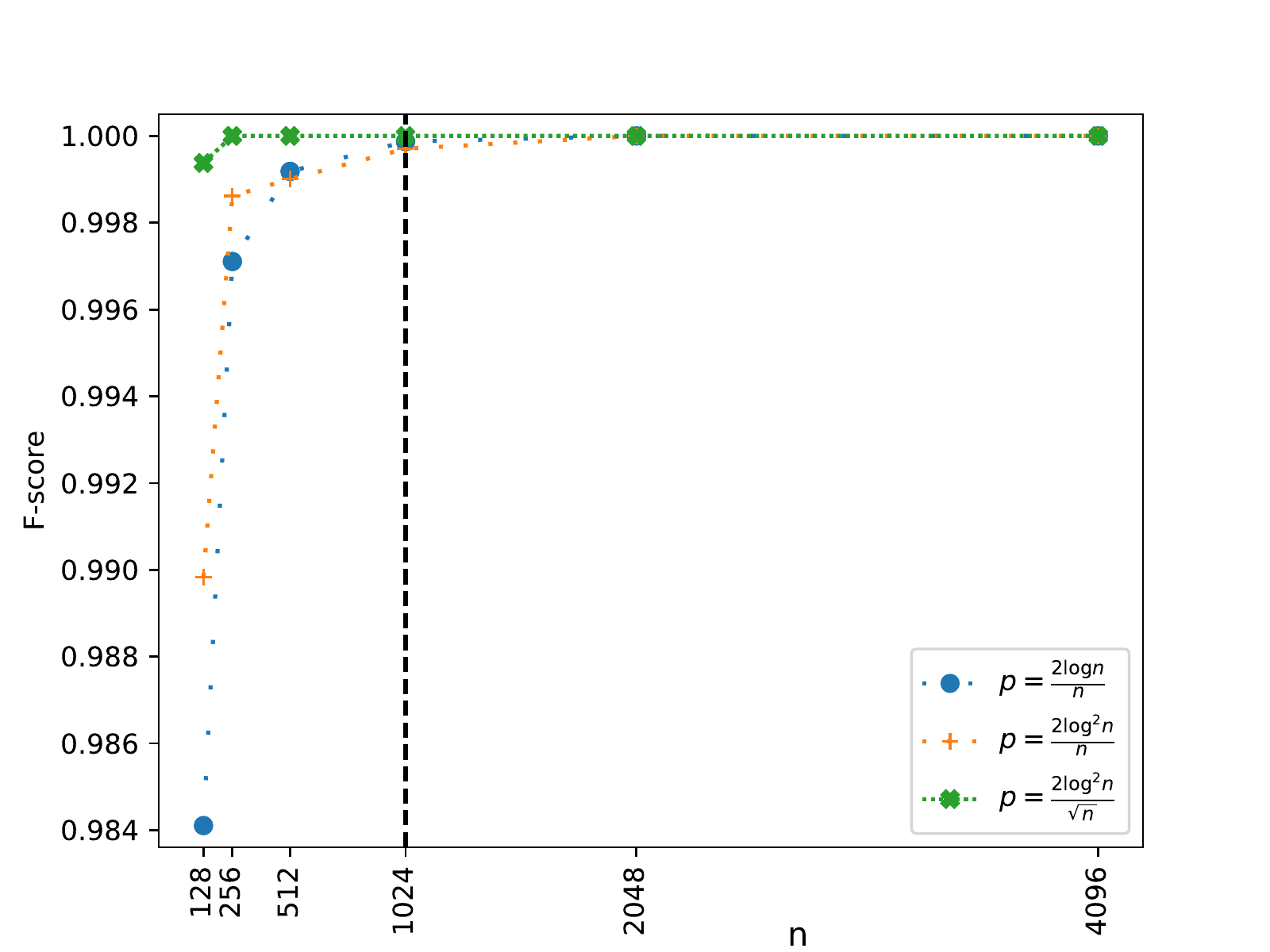}
	%\caption{\footnotesize Accuracy of CDRW algorithm on $G_{np}$ random graphs. It shows the performance of CDRW on  sparse random graphs. Even when the graph is sparse, as close to the connectivity threshold as possible, its accuracy is still high. The vertical line shows that when the size is big enough ($n\geq 2^{10}$), the accuracy becomes almost $1.0$.}
	\caption{\footnotesize Community detection accuracy of CDRW algorithm on $G_{np}$ random graphs. It shows even when the graph is sparse, when $p$ is small and as close to the connectivity threshold as possible, its accuracy is still high. The vertical line shows that when the size is big enough ($n\geq 2^{10}$), the accuracy becomes almost $1.0$.}
	\label{fig:gnp}
\end{figure}

The first challenge for any community detection (CD) algorithm is detecting a random graph as a single community. This challenge becomes harder when the graph becomes sparse and it gets closer to the connectivity threshold of a random graph (i.e. $p=\frac{c\log n}{n}$, s.t. $c>1$) \cite{bollobas1998random}. In the first experiment we show that our CDRW algorithm detects almost the whole graph as a single community resulting in a high F-score accuracy value, see Figure~\ref{fig:gnp}. Figure \ref{fig:gnp} shows that when we increase the size of graph $n$, the accuracy of our algorithm increases as well. For example, for $n=2^{10}$ the accuracy metric becomes almost $1.0$, meaning that almost all the nodes of the graph is detected as a single community. It also shows that when $p$ increases (graph gets denser), the accuracy also increases. So in the remaining experiments on PPM graphs, we choose two lowest values of $p=\frac{c\log n}{n}$ and $p=\frac{c\log^2{n}}{n}$ for generating its random parts in order to give more challenging input graphs to the CDRW algorithm.

 \begin{figure}[t]
 	\centering
 	\includegraphics[width=0.5\textwidth]{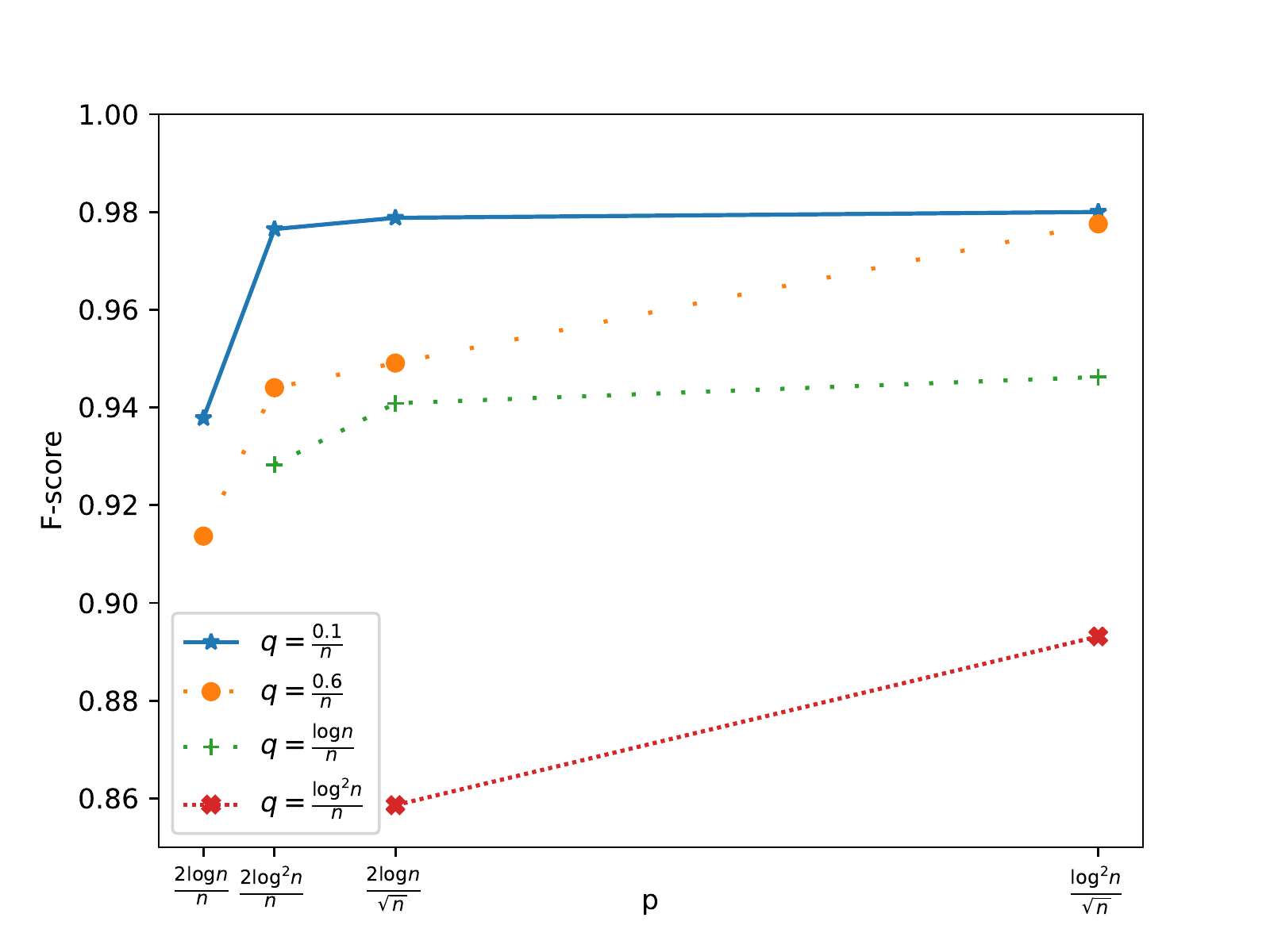}
 	\caption{\footnotesize Performance of CDRW algorithm on PPM graphs when there are two parts/communities ($r=2$). We fixed the size of the graph to $n=2^{11}$, each planted partition is of size $2^{10}$. It shows that CDRW works well for small values of $p=\frac{2\log n}{n}$ and $p=\frac{2\log^2{n}}{n}$ when $q$ is small enough.}
 	\label{fig:gnpqk2}
\end{figure}

%\begin{figure}[t]
%	\centering
%	\includegraphics[width=0.5\textwidth, height=0.28\textwidth]{./parts/figs2/plotoverkonefig.pdf}
%	\caption{\footnotesize Varying the number of ground-truth communities to see its effect on the accuracy of our CDRW algorithm. It shows that when we increase the number of communities, the accuracy decreases slightly. This is expected because the number of inter-community edges increases. Comparing Sub-figure \ref{fig:nk1} and \ref{fig:nk2}, we see that if we fix the number of communities, then the accuracy gets higher when the  size of the communities becomes larger.}
%	\label{fig:onefig}
%\end{figure}

After showing that CDRW works well on $G_{np}$ random graphs, now we consider PPM $G_{npq}$ graphs. At first we fix the number of communities to two ($r=2$) so that we can consider the effect of various values of $p$ and $q$. This will show us the threshold for the ratio of $\frac{p}{q}$ where CDRW works well. As we showed in Figure \ref{fig:gnp}, when the size of each random graph is big enough ($n\geq 2^{10}$), CDRW detects a single $G_{np}$ community well. Therefore we set the size of $G_{npq}$ to $n=2^{11}$ which makes each ground-truth community big enough ($\frac{n}{r}=2^{10}$). 
When considering PPM graphs with $p$ and $q$, as the connectivity probability for intra- and inter-community edges, CD algorithms face hardship in detecting communities when  $p$ is small and $q$ is relatively high. 
%In other words, community detection is a hard task when the communities are sparse and the inter-community edges are dense. 
But the $\frac{p}{q}$ ratio can not be arbitrarily small because it causes the two communities blend into each other and the graph looses its community structure.
Figure \ref{fig:gnpqk2} shows accuracy of CDRW for different values of $p$ and $q$. We highlight that it shows even for sparse parted $G_{npq}$ graphs: for $p=\frac{2\log n}{n}$, CDRW detects the two communities with a high F-score value (more than $0.90$) for $q=\frac{0.1}{n}$ and $\frac{0.6}{n}$. In other words, our CDRW algorithm works well even on sparse parted PPM graphs when the $\frac{p}{q}$ ratio is as small as ($\Omega(\log n)$). Notice that when $p=\frac{2\log n}{n}$, the two ground truth communities of the PPM graph are as sparse as possible, i.e close to its connectivity threshold. In the latter example, for instance, when $q=\frac{0.6}{n}$, a partition has in expectation $e_{in} = \binom{\frac{n}{r}}{2} p = 10230$ intra and $e_{out}= \frac{n}{r}(n-\frac{n}{r})q = 614$ inter community edges. It means the ratio of inter to intra community edges ($\frac{e_{out}}{e_{in}}$) is high and equal to $6\%$.

\begin{figure*}
	\begin{subfigure}[]{\columnwidth}
		\includegraphics[width=\textwidth, left]{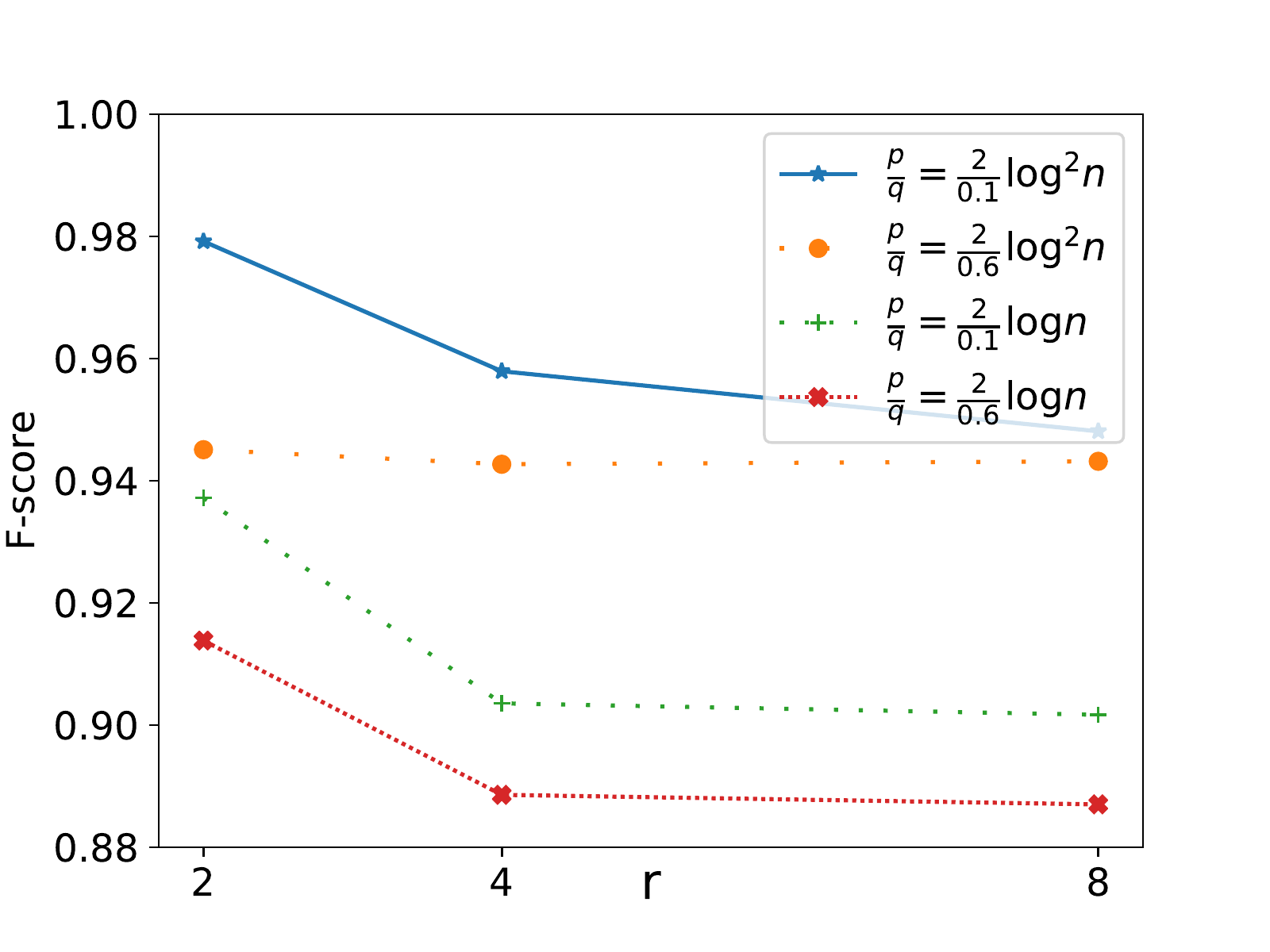}
		\caption{$n = r \times 2^{10}$.}
		\label{fig:nk1}
	\end{subfigure}%
	\begin{subfigure}[]{\columnwidth}
		\includegraphics[width=\textwidth, left]{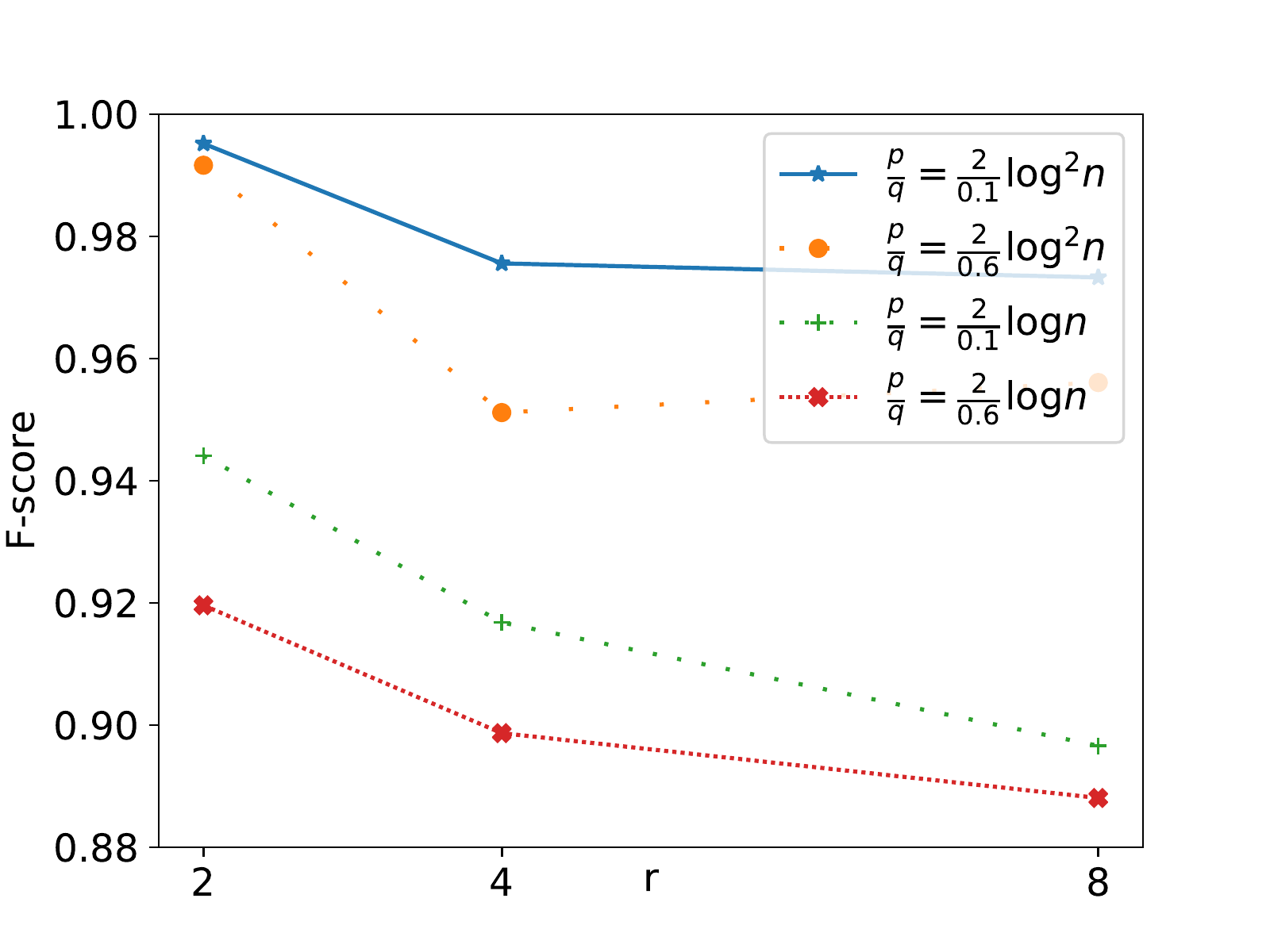}
		\caption{$n=8 \times 2^{10}$.}
		\label{fig:nk2}
	\end{subfigure}
	\caption{\footnotesize Varying the number of ground-truth communities to see its effect on the accuracy of our CDRW algorithm. It shows that when we increase the number of communities, the accuracy decreases slightly. This is expected because the number of inter-community edges increases. Comparing Sub-figure \ref{fig:nk1} and \ref{fig:nk2}, we see that if we fix the number of communities, then the accuracy gets higher when the  size of the communities becomes larger.}
	\label{fig:diffk}
\end{figure*}

We now consider the effect of increasing the number of ground-truth communities ($r$) in order to see its effect on the accuracy of our CDRW algorithm, see Figure~\ref{fig:diffk}. We do it in two ways. First, we fix the size of each community to $2^{10}$ and vary the number of communities. The size of graph is  $n=r\times2^{10}$. Figure \ref{fig:nk1} shows that our CDRW algorithm works well when we reasonably increase the number of communities. Second, we fix the size of graph to a number so that the size of each community is $2^{10}$ when the number of communities is the biggest ($r=8$), see Figure \ref{fig:nk2}. Then, when the number of communities becomes lower, the size of communities gets bigger. By comparing Sub-figures \ref{fig:nk1} and \ref{fig:nk2}, we see that when the number of  communities are the same, the accuracy is higher when the size of each community is bigger.

\vspace{-0.05in}
\section{Conclusion}
\vspace{-0.05in}
We proposed a distributed algorithm, CDRW, for community detection that works well for the PPM model ($G_{npq}$ random graph), a standard random graph model used extensively in clustering and community detection studies.  Our CDRW algorithm is relatively simple and localized; it uses random walk and mixing time property of graphs  to detect communities. We provide a rigorous 
theoretical analysis of our algorithm on the $G_{npq}$ random graph and characterize its performance vis-a-vis the parameters of the model. In particular, our main result is that it correctly identifies the communities provided $q = o(p/(r\log(n/r)))$, where $r$ is the number of communities. Our algorithm takes $O(r \times \mbox{polylog } n)$ rounds and hence
is quite fast when $r$ is relatively small. We point out that our algorithm can also be extended to find
communities even faster (by finding communities in parallel), assuming we know an (estimate) of $r$. 
For future work, it will be interesting to study the performance of this algorithm on other graph models and can be a starting point to design and analyze  community detection algorithms that perform well in the more challenging case of real-world graphs.
%both sparse and dense synthesis graphs while degree distribution of nodes are close to uniform. 

\iffalse
Our second method called CDR algorithm works well on both synthesis and real world graphs. Although the former does not work well on real world graphs, it beats the latter one on sparse parted $PPM$ graphs. This shows that it has potential for further investigation as a future work.

\par
Our experiments included different size of graphs, community sizes, number of communities and density of inter and intra community densities. Our proposed algorithms over performed compared base line methods in the quality of detected communities. We also showed our CDR algorithm runs more efficient in terms of time, message and space complexity compared to baseline methods. Providing analytical proof of correctness for the CDR algorithm would be also an interesting further investigation.
\fi

\balance
\bibliographystyle {abbrv}
\bibliography{./parts/main,biblio,pram}
\end{document}